\documentclass[letterpaper, 10 pt, conference]{ieeeconf}

\IEEEoverridecommandlockouts                              

\usepackage{graphicx} 

\usepackage{amsmath} 
\usepackage{dsfont}
\usepackage{algpseudocode}
\usepackage{algorithm}
\usepackage{mathtools}
\usepackage[scr]{rsfso}
\usepackage{array, amssymb, amscd}
\usepackage{url}
\usepackage[percent]{overpic}
\usepackage{cite} 

\usepackage{color}
\usepackage{xcolor}
\usepackage{tikz}
\usetikzlibrary{arrows, backgrounds, calc}

\definecolor{lightred}{RGB}{241, 225, 222}
\definecolor{lightred2}{RGB}{241, 225, 222}
\definecolor{color1}{rgb}{0.1,0.498039215686275,0.9549019607843137}
\definecolor{alizarin}{rgb}{0.82, 0.1, 0.26}
\definecolor{antiquewhite}{rgb}{0.98, 0.92, 0.84}
\definecolor{azure}{rgb}{0.94, 1.0, 1.0}
\definecolor{offwhite}{rgb}{0.98, 0.95, 0.95}
\definecolor{pigment}{rgb}{0.2, 0.2, 0.6}
\definecolor{BrickRed}{rgb}{0.8, 0.25, 0.33}

\usepackage{hyperref}

\usepackage{todonotes}

\newtheorem{lemma}{Lemma}
\newtheorem{thm}{Theorem}
\newtheorem{corollary}{Corollary}

\newcommand{\R}{\mathbb{R}}

\newcommand{\1}{\mathds{1}}

\title{\LARGE \bf
    Certified Approximate Reachability (CARe):\\ Formal Error Bounds on Deep Learning of Reachable Sets}

\author{Prashant Solanki$^{*1}$, Nikolaus Vertovec$^{*2}$, Yannik Schnitzer$^{2}$,\\
        Jasper Van Beers$^{1}$, Coen de Visser$^{1}$, Alessandro Abate$^{2}$
\thanks{*Authors contributed equally to this article.}
\thanks{$^{1}$Department of Aerospace Engineering, TU Delft, Netherlands
        {\small \{p.solanki, j.j.vanbeers, c.c.devisser\}@tudelft.nl}}%
\thanks{$^{2}$Department of Computer Science, University of Oxford, OX1 3PJ, UK
        {\small \{nikolaus.vertovec, yannik.schnitzer, alessandro.abate\}@cs.ox.ac.uk}}%
}

\begin{document}
\maketitle

\begin{abstract}    
Recent approaches to leveraging deep learning for computing reachable sets of continuous-time dynamical systems have gained popularity over traditional level-set methods, as they overcome the curse of dimensionality. However, as with level-set methods, considerable care needs to be taken in limiting approximation errors, particularly since no guarantees are provided during training on the accuracy of the learned reachable set. To address this limitation, we introduce an \(\epsilon\)-approximate Hamilton-Jacobi Partial Differential Equation (HJ-PDE), which establishes a relationship between training loss and accuracy of the \textit{true} reachable set. To formally certify this approximation, we leverage Satisfiability Modulo Theories (SMT) solvers to bound the residual error of the HJ-based loss function across the domain of interest. Leveraging Counter Example Guided Inductive Synthesis (CEGIS), we close the loop around learning and verification, by fine-tuning the neural network on counterexamples found by the SMT solver, thus improving the accuracy of the learned reachable set. To the best of our knowledge, Certified Approximate Reachability (CARe) is the first approach to provide soundness guarantees on learned reachable sets of continuous dynamical systems.
\end{abstract}

\newif\iflongversion

\longversiontrue

\section{Introduction}
Ensuring the safety of autonomous systems under uncertainty is a fundamental challenge in control and verification. Reachability analysis plays a key role in addressing safety concerns across various domains, including spacecraft trajectory design~\cite{Vertovec2021, Vertovec2023}, ground transportation systems~\cite{livadas1998formal, lygeros1998verified}, air traffic management~\cite{livadas2000high, tomlin2001safety}, and flight control~\cite{Vertovec2024, lygeros2004reachability, lygeros1999controllers}. The theoretical foundations of reachability stem from viability theory~\cite{aubin2011viability, aubin1991viability}, while computational techniques have been developed for both exact and approximate reachable set computations in hybrid systems~\cite{cardaliaguet1996differential, cardaliaguet2007differential, asarin2000effective, mitchell2007toolbox, mitchell2001validating}. These advances have formulated reachability analysis as an optimal control problem, where reachable, viable, or invariant sets are represented as level sets of a value function satisfying a Hamilton-Jacobi (HJ) partial differential equation (PDE)~\cite{bansal2017hamilton, chen2021fastrack, chen2018decomposition, lygeros2004reachability, maidens2013lagrangian}. 

Traditionally, level set methods are used to compute the unique viscosity solution of the HJ equation~\cite{mitchell2007toolbox}. However, even with high-order numerical schemes such as WENO methods, the non-differentiability of the value function introduces numerical inaccuracies~\cite{Vertovec2022}. While refining the grid used in finite element methods such as level set methods improves numerical accuracy, it is hard to predict what constitutes a sufficiently fine grid. Moreover, these grid-based approaches suffer from the curse of dimensionality, as the number of required grid points grows exponentially with system dimensionality.

Deepreach~\cite{bansal2021deepreach}, a recent approach leveraging neural networks, has enabled solving the HJ equation without relying on finite difference methods. If the loss function enforcing the PDE conditions converges to zero in the region of interest, the learned value function is the unique viscosity solution of the HJ-PDE, ensuring that its zero-level set correctly represents the reachable set. This approach improves scalability to higher-dimensional systems by mitigating the curse of dimensionality.

In this work, we demonstrate an additional advantage of using neural networks to approximate the value function: the ability to provide a formally bounded \(\epsilon\)-accurate reachable set. While both level-set methods and DeepReach yield empirically accurate solutions, their adherence to the HJ-PDE conditions--and thus the validity of the reachable set--has not been quantified. Follow-up work to DeepReach uses probabilistic methods to recover the reachable set \cite{Lin_2023, Lin_2024, Tayal2025}. In contrast, we provide a formal approach to verification that does not rely on recovery of the reachable set, but rather provides a sound over- and under-approximation of the reachable set. Our key contributions are: 
\begin{enumerate}
    \item We establish a formal bound \(\epsilon\) on the residual error, ensuring that the absolute value of the HJ-based loss function remains within a specified threshold across the domain of interest.  

    \item We demonstrate that the neural network-based value function, constrained by \(\epsilon\), provides a certified approximation of the true value function, allowing for over- and under-approximation of the reachable set.  
\end{enumerate}
The remainder of the paper is structured as follows: Section \ref{sec2:problem_setup} provides an overview of the reachability problems and introduces the Hamilton-Jacobi (HJ) partial differential equation (PDE) whose solution we aim to learn. Section \ref{sec3:deep_learning} presents the deep learning approach for computing the reachable set, while Section \ref{sec4:smt} discusses the formal verification method used to derive and certify an upper bound on the loss across the entire domain of interest. In Section \ref{sec5:main_theorem}, we introduce our main contribution, establishing a connection between the residual training loss and an under- and over-approximation of the reachable set. Finally, Section \ref{sec6:case_study} provides implementation details and a case study.

\section{Problem Setup} \label{sec2:problem_setup}
In reachability theory, a key objective is to compute the backward reachable set (BRS) of a dynamical system, which consists of all states from which trajectories can reach a given target set at the end of the time horizon. In contrast, the backward reachable tube (BRT) represents the set of states that can reach the target set within a specified time horizon~\cite{bansal2017hamilton}. If the target set represents unsafe states, the BRS/BRT identifies states that may lead to unsafe conditions and should be avoided. To account for adversarial disturbances, safety-critical scenarios are often modeled as a two-player game, where Player 1 represents the control input and Player 2 represents the disturbance input~\cite{margellos2011hamilton}.

Mathematically, consider a dynamical system with state \( x \in \mathbb{R}^{m} \) governed by the ordinary differential equation (ODE)  
\begin{equation}\label{eq_main_ODE}
    \dot{x}(s) \in f(s,x(s),u(s),d(s)), \quad t_{0} \leq s \leq T,
\end{equation}
where \( T \geq t_{0} \) is the fixed time horizon. The initial state is given by \( x(t_{0}) \coloneqq x_{0} \), and the control and disturbance inputs are measurable functions
\begin{align*}
    u:[t_{0},T] \rightarrow \mathcal{U}, \quad d:[t_{0},T] \rightarrow \mathcal{D},
\end{align*}
with \( \mathcal{U} \subset \mathbb{R}^{k} \) and \( \mathcal{D} \subset \mathbb{R}^{l} \) being compact sets. The dynamics \( f:[t_{0},T] \times \mathbb{R}^{m} \times \mathcal{U} \times \mathcal{D} \rightarrow \mathbb{R}^{m} \) satisfy the following conditions for some constants \( C_{1}, C_{2} \):  
\begin{align}
    |f(t,x,u,d)| &\leq C_{1}, \label{eq_main_system_assumtion_1} \\
    |f(t,x,u,d)-f(t,\hat{x},u,d)| &\leq C_{2} |x - \hat{x}|, \label{eq_main_system_assumtion_2}
\end{align}
for all \( t \in [t_{0},T] \), \( x,\hat{x} \in \mathbb{R}^{m} \), \( u \in \mathcal{U} \), and \( d \in \mathcal{D} \). These assumptions ensure that the ODE \eqref{eq_main_ODE} admits a unique solution  
\begin{equation}\label{eq_main_solution_ODE}
    x(t) = \phi(t,t_{0},x_{0},u(\cdot),d(\cdot)).
\end{equation}  
Next, we define the sets of control and disturbance policies:  
\begin{align*}
    \mathcal{M}_{[t,T]} &\equiv\{u:[t,T]\rightarrow \mathcal{U} \mid u \text{ is measurable} \}, \\
    \mathcal{N}_{[t,T]} &\equiv\{d:[t,T]\rightarrow \mathcal{D} \mid d \text{ is measurable} \}.
\end{align*}
A nonanticipative strategy is a mapping \( \beta: \mathcal{M}_{[t,T]} \rightarrow \mathcal{N}_{[t,T]} \) such that for all \( s \in [t,T] \) and for all \( u, \hat{u} \in \mathcal{M}_{[t,T]} \), if \( u(\tau) = \hat{u}(\tau) \) for almost every \( \tau \in [t,s] \), then \( \beta[u](\tau) = \beta[\hat{u}](\tau) \) for almost every \( \tau \in [t,s] \). The class of such strategies is denoted as \( \Delta_{[t,T]} \).

Given a target set \( \mathcal{G} \) with a signed distance function \( g(x) \) such that \( g(x) \leq 0 \iff x \in \mathcal{G} \), the BRT is defined as  
\begin{align}\label{ch3:eqn:reachavoidtube}
	\mathrm{BRT}_{\mathcal{G}}([t_{0},T]) =& \big\{ x \mid \exists \beta \in \mathcal{N}_{[t_{0},T]},
        \forall u \in  \mathcal{M}_{[t_{0},T]}, \\ \nonumber
	&\exists s \in [t_{0},T]: \phi(s,t_{0},x_{0},u(\cdot),\beta(\cdot)) \in \mathcal{G} \big\}.
\end{align}
Equivalently, the BRT is given by the subzero level set of the value function \( V^{*}(t,x) \), where  
\begin{align}\label{eq_main_value_func_star}
   V^{*}(t,x) = \inf_{\beta \in \Delta_{[t,T]}} \sup_{u \in \mathcal{M}_{[t,T]}} \inf_{\tau \in [t,T]} g(\phi(\tau,t,x,u(\cdot),\beta(\cdot))).
\end{align}
Let us introduce the shorthand notation \( D_{x}V(t,x) = \frac{\partial V(t,x)}{\partial x} \) and \( D_{t}V(t,x) = \frac{\partial V(t,x)}{\partial t} \). Then, a standard result of HJ reachability~\cite{evans1984differential, margellos2011hamilton, bansal2017hamilton} allows us to reduce the optimization problem in Equation~\eqref{eq_main_value_func_star} from an infinite-dimensional optimization over policies to an optimization over control and disturbance inputs by formulating the value function as the viscosity solution of the HJ-PDE:
\begin{align}\label{eq_hjb_Pde_star}
     D_{t}V^{*}(t,x) +  \min \{ 0 , \mathcal{H}(t,x, D_{x}V^{*})\} = 0, \nonumber \\ 
    V^{*}(T,x) = g(x), 
\end{align}
with the Hamiltonian defined as  
\begin{align*}
    \mathcal{H}(t,x,p) = \max_{u \in \mathcal{U}} \min_{d \in \mathcal{D}} [f(t, x, u, d) \cdot p].
\end{align*}
The HJ-PDE can be solved using level-set methods~\cite{mitchell2001validating, tomlin2000game, lygeros2004reachability} or through curriculum training of a neural network~\cite{bansal2021deepreach}. 

\section{Deep Learning Approach to HJ Reachability}\label{sec3:deep_learning}
Inspired by the self-supervised framework of Physics-Informed Neural Networks (PINNs)~\cite{Raissi_2019}, the solution of the PDE~\eqref{eq_hjb_Pde_star} can be effectively learned, as demonstrated in~\cite{bansal2021deepreach}. 

Let \( V_{\theta} \) denote the approximate value function learned by a neural network. To ensure that the zero level set of \( V_{\theta} \) correctly represents the BRT, the network must achieve zero loss for the loss function defined as
\begin{align}\label{eq_deepreach_loss_func}
    \mathcal{L}(t_{i}, x_{i};\theta) &= \1(t_{i} = T)\mathcal{L}_{1}(x_{i};\theta) + \lambda \mathcal{L}_{2}(t_{i}, x_{i};\theta), \\
    \mathcal{L}_{1}(x_{i};\theta) &= \|V_{\theta}(T, x_{i}) - g(x_{i})\|, \\
    \mathcal{L}_{2}(t_{i}, x_{i};\theta) &= \|D_{t} V_{\theta} (t_{i}, x_{i}) + \min[0,\mathcal{H}(t_{i}, x_{i}, D_{x}V_{\theta})]\|,
\end{align}
where \( \1(t_{i} = T) \) is an indicator function ensuring that \( \mathcal{L}_{1} \) is only applied at \( t_{i} = T \), and \( \lambda \) controls the balance between the two loss terms. The nonlinear nature of the neural network might result in areas of the state space with significantly higher loss than empirically observed during training. This acute shortcoming necessitates a formal approach to verifying the approximation errors. To this end, in the subsequent sections, we will
\begin{enumerate}
    \item establish a verification method to ensure \( |\mathcal{L}_1(x;\theta)| < \epsilon_1 \) and \( |\mathcal{L}_2(t,x;\theta)| < \epsilon_2 \) for all \( x \) and \( t \) within the domain of interest; and 
    \item derive upper and lower bounds on \( V_{\theta} \) in terms of viscosity solutions to the HJ-PDE, thus enabling under- and over-approximations of the BRT/BRS.
\end{enumerate}

\section{Formal Synthesis of Neural Value Functions} \label{sec4:smt}
Our approach to synthesizing the value function solution follows a two-phase process--a learning phase and a certification phase--that alternates in a Counter Example Guided Inductive Synthesis (CEGIS) loop~\cite{Abate_2018}. The learning phase employs a curriculum training approach, discussed in detail in Section \ref{sec6:case_study}, terminating when we empirically satisfy \( |\mathcal{L}_1(x_{i};\theta)| < \epsilon_1 \) and \( |\mathcal{L}_2(t_{i}, x_{i};\theta)| < \epsilon_2 \), for \(N\) randomly sampled points \( (t_1, x_1), \ldots, (t_n, x_n) \).

To further ensure that \( |\mathcal{L}_1(x;\theta)| < \epsilon_1 \) and \( |\mathcal{L}_2(t, x;\theta)| < \epsilon_2 \) for all \( x \) in the domain of interest \( \mathcal{X} \) and for all \( t \in [t_{0}, T] \), we design a sound certification phase: this certification is performed via SMT (Satisfiability Modulo Theories) solving, which symbolically reasons over the continuous domain \( [t_{0}, T] \times \mathcal{X} \), hence generalizing across the sample-based loss error. While computationally intensive, formal verification provides soundness guarantees, allowing us to rigorously validate the learned value function as an approximate solution to the HJ-PDE.

Several approaches exist for certifying the accuracy of the learned value function. To remain general in our choice of dynamical systems, we employ an SMT solver capable of handling quantifier-free nonlinear real arithmetic formulae~\cite{Franzle_2007}. This allows us to incorporate arbitrary nonlinear activation functions into the neural network. Furthermore, it enables us to address systems with nonlinear dynamics, which induce nonlinearities in the Hamiltonian. Specifically, we employ \texttt{dReal}~\cite{Gao_2013}, which supports both polynomial and non-polynomial terms, such as transcendental functions (e.g., trigonometric and exponential functions).

After generating a symbolic representation of the neural network, the SMT solver searches for an assignment of variables \( (t, x) \) that satisfies the quantifier-free formula
\begin{align}
\big(& x \in \mathcal{X} \land \|V_{\theta}(T, x)-g(x)\| > \epsilon_1 \big) \label{formula_1} \; \lor \\
\big(& x \in \mathcal{X} \land t \in [t_{0}, T] \nonumber \; \land \\
   \| & D_{t} V_{\theta}(t, x) + \min[0,\mathcal{H}(t, x, D_{x}V_{\theta})]\| > \epsilon_2 \big)
   \raisetag{-8pt}\label{formula_2}
\end{align}
The logical disjunction between \eqref{formula_1} and \eqref{formula_2} allows us to split the SMT call into separate parallel queries. If one query identifies a valid assignment, i.e., a counterexample to either \( |\mathcal{L}_1(x;\theta)| < \epsilon_1 \) or \( |\mathcal{L}_2(t, x;\theta)| < \epsilon_2 \), the remaining SMT calls terminate. A counterexample indicates the necessity to finetune the neural network, as it does not yet approximate a valid value function with sufficient accuracy. We leverage the counterexample for targeted training specifically at inputs where the network is insufficiently accurate. We sample additional training data points for finetuning in close proximity to the counterexample, thus reducing the need for extensive sampling in regions of the input space that are already well approximated. Finetuning and certification are alternately repeated within an inductive synthesis loop (CEGIS), progressively improving the neural value function in targeted regions of the input space, until the certifier determines the Formulae \eqref{formula_1} and \eqref{formula_2} to be unsatisfiable. This implies that there exists no counterexample to \( |\mathcal{L}_1(x;\theta)| < \epsilon_1 \) and \( |\mathcal{L}_2(t, x;\theta)| < \epsilon_2 \), formally proving the validity of the learned value function over the entire domain of interest. In the following section, we derive results establishing how this formal certification can be leveraged to obtain sound over- and under-approximations of the BRT and BRS from the neural value function.


\section{\texorpdfstring{$\epsilon$}{ε}-Accurate Value Functions} \label{sec5:main_theorem}

The certification of the neural value function in the previous section established that for all \( (t, x) \in [t_{0}, T] \times \mathcal{X} \), there exists some \( \hat{\epsilon} \in [0, \epsilon_2] \) such that
\begin{align*}
    \left\| D_{t} V_{\theta}(t, x) + \min[0,\mathcal{H}(t, x, D_{x}V_{\theta})] \right\| = \hat{\epsilon},
\end{align*}
which can equivalently be rewritten as
\begin{align*}
    \left\| D_{t} V_{\theta}(t, x) + \min[\pm \hat{\epsilon},\tilde{\mathcal{H}}(t, x, D_{x}V_{\theta})] \right\| = \hat{\epsilon}, \nonumber \\
    \tilde{\mathcal{H}}(t, x, p) = \max_{u \in \mathcal{U}} \min_{d \in \mathcal{D}} \left[ f(t, x, u, d) \cdot p \pm\hat{\epsilon} \right].
\end{align*}
We consider the worst-case realisation of \(\hat{\epsilon}\), i.e., \(\hat{\epsilon} = \epsilon_2\). We will now show that, for \(\epsilon = (\epsilon_1 + \epsilon_2 (T-t))\),
\begin{equation} \label{eq_value_func_star_nn_bounded}
    V_{\theta}(t,x) - \epsilon\leq V^{*}(t,x) \leq V_{\theta}(t,x) + \epsilon,
\end{equation} 
which allows us to use the neural value function to reason about the true BRT and BRS.

\subsection{Bounding the True Value Function}
To derive these bounds, we introduce the modified value functions \( \overline{V} \) and \( \underline{V} \), which are the unique viscosity solutions under the worst-case realizations of \( \hat{\epsilon} \):
\begin{align}
   \overline{V}(t,x) &= \inf_{\beta(\cdot) \in \Delta_{[t,T]}} \sup_{u(\cdot) \in \mathcal{M}_{[t,T]}} \nonumber \\
   &\quad \inf_{\tau \in [t,T]} \left[ \int_{t}^{\tau} \epsilon_2 \, ds + g(\phi(\tau,t,x,u(\cdot),\beta(\cdot))) \right], \label{eq_main_value_func} \\
   \underline{V}(t,x) &= \inf_{\beta(\cdot) \in \Delta_{[t,T]}} \sup_{u(\cdot) \in \mathcal{M}_{[t,T]}} \nonumber \\
   &\quad \inf_{\tau \in [t,T]} \left[ \int_{t}^{\tau} -\epsilon_2 \, ds + g(\phi(\tau,t,x,u(\cdot),\beta(\cdot))) \right]. \label{eq_main_value_func_negative}
\end{align}

Since \( \epsilon_2 \) is a positive constant, it follows that
\begin{equation}
    \underline{V}(t,x) \leq V^*(t,x) \leq \overline{V}(t,x).
\end{equation}
Furthermore, using the inequality \( -\epsilon_2 (T-t) \leq \int_{t}^{\tau} -\epsilon_2 \, ds \), we obtain:
\begin{align*}
    &\inf_{\tau \in [t,T]} g(\phi(\tau,t,x,u(\cdot),\beta(\cdot))) - \epsilon_2 (T-t) \nonumber \\
    & \leq \inf_{\tau \in [t,T]} \left[ g(\phi(\tau,t,x,u(\cdot),\beta(\cdot))) - \int_{t}^{\tau} \epsilon_2 \, ds \right].
\end{align*}
Similarly, using \( \epsilon_2 (T-t) \geq \int_{t}^{\tau} \epsilon_2 \, ds \), we obtain:
\begin{align*}
    &\inf_{\tau \in [t,T]} g(\phi(\tau,t,x,u(\cdot),\beta(\cdot))) + \epsilon_2 (T-t) \nonumber \\
    & \geq \inf_{\tau \in [t,T]} \left[ g(\phi(\tau,t,x,u(\cdot),\beta(\cdot))) + \int_{t}^{\tau} \epsilon_2 \, ds \right].
\end{align*}
Thus, we derive the key bound:
\begin{equation} \label{eqn:V_bnds}
    \overline{V}(t,x) - \epsilon_2(T-t) \leq V^{*}(t,x) \leq \underline{V}(t,x) + \epsilon_2(T-t).
\end{equation}

\subsection{Relating Bounds to the Neural Value Function}
To relate \( \overline{V} \) and \( \underline{V} \) to the neural value function, we need to establish that these auxiliary value functions are the unique viscosity solutions of the \( \epsilon \)-modified HJ-PDE. This result is formalized in the following theorem.

\begin{thm}\label{theorem_viscosity_sol_hjiPde}
    The function \( \overline{V}(t,x) \) is the unique viscosity solution of the \( \epsilon \)-modified HJ-PDE:
    \begin{align*}
        D_{t} \overline{V}(t,x) +  \min\left\{ \epsilon_2 , \mathcal{H}(t,x, D_{x}\overline{V})\right\} = 0,
    \end{align*}
    with the Hamiltonian defined as
    \begin{align*}
        \mathcal{H}(t, x, p) = \max_{u \in \mathcal{U}} \min_{d \in \mathcal{D}} \left[ f(t, x, u, d) \cdot p + \epsilon_2 \right],
    \end{align*}
    and the terminal condition
    \begin{align*}
        \overline{V}(T,x) = g(x).
    \end{align*}
    
    Similarly, \( \underline{V}(t,x) \) is the unique viscosity solution of the \( \epsilon \)-modified HJ-PDE:
    \begin{align*}
        D_{t} \underline{V}(t,x) +  \min\left\{ -\epsilon_2 , \mathcal{H}(t,x, D_{x}\underline{V})\right\} = 0,
    \end{align*}
    with the Hamiltonian
    \begin{align*}
        \mathcal{H}(t, x, p) = \max_{u \in \mathcal{U}} \min_{d \in \mathcal{D}} \left[ f(t, x, u, d) \cdot p - \epsilon_2 \right],
    \end{align*}
    and the terminal condition
    \begin{align*}
        \underline{V}(T,x) = g(x).
    \end{align*}
\end{thm}
\iflongversion
    The proof of Theorem \ref{theorem_viscosity_sol_hjiPde} is given in the appendix.
\else
    The proof of Theorem \ref{theorem_viscosity_sol_hjiPde} is given in the extended version of the paper \cite{}.
\fi

As a consequence of Theorem~\ref{theorem_viscosity_sol_hjiPde}, we obtain
\begin{align} \label{eqn:Vtheta_bnds}
    \underline{V}(t,x) \leq V_{\theta}(t,x) \leq \overline{V}(t,x).
\end{align}
Finally, incorporating \( \epsilon_1 \), which perturbs the boundary condition, we have:
\begin{align}\label{eqn:Vtheta_boundary}
    V_{\theta}(T,x) = g(x) \pm \epsilon_1.
\end{align}
Thus, by combining Equation~\eqref{eqn:V_bnds} with Equations~\eqref{eqn:Vtheta_bnds} and~\eqref{eqn:Vtheta_boundary}, we recover the bound in Equation~\eqref{eq_value_func_star_nn_bounded}.
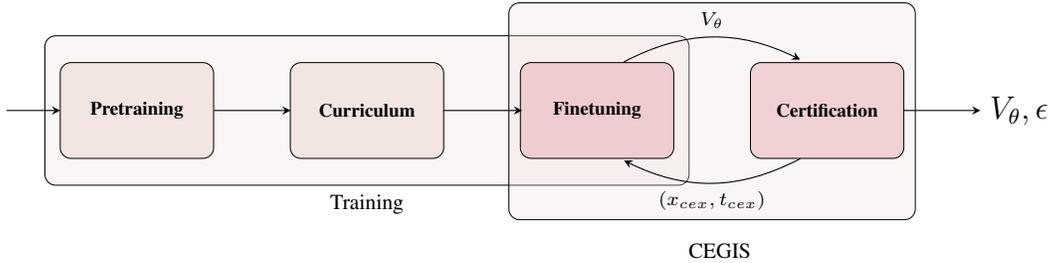
\begin{figure*}[h]
\vspace{-5pt}
\centering
\resizebox{0.8\textwidth}{!}{%
\begin{tikzpicture}

\begin{pgfonlayer}{background}
  \draw [fill=lightred, line width=0.2pt, fill opacity=0.3, text opacity=1, rounded corners] 
    (2.3,10.6) rectangle (10.7,8.65);
  \node at (6.5,8.4) {\footnotesize Training};
\end{pgfonlayer}

\begin{scope}[on background layer]
\end{scope}

\draw [fill=lightred, line width=0.2pt, fill opacity=0.3, text opacity=1, rounded corners] 
    (8.35,11.03) rectangle (13.65,8.2);

\draw [fill=BrickRed, line width=0.2pt, fill opacity=0.2, text opacity=1, rounded corners] 
    (8.5,10.25) rectangle (10.5,9) node[pos=.5] (learner) {\scriptsize \textbf{Finetuning}};

\draw [fill=BrickRed, line width=0.2pt, fill opacity=0.2, text opacity=1, rounded corners] 
    (11.5,10.25) rectangle (13.5,9) node[pos=.5] (certifier) {\scriptsize \textbf{Certification}};

\draw [fill=lightred, line width=0.2pt, fill opacity=0.8, text opacity=1, rounded corners, align=center] 
    (2.5,10.25) rectangle (4.5,9) node[pos=.5] (pretraining) {\scriptsize \textbf{Pretraining}};

\draw [fill=lightred, line width=0.2pt, fill opacity=0.8, text opacity=1, rounded corners, align=center] 
    (5.5,10.25) rectangle (7.5,9) node[pos=.5] (subquo) {\scriptsize \textbf{Curriculum}};

\draw[-stealth] ([xshift=0.8em]pretraining.east) -- ([xshift=-0.7em]subquo.west);

\draw[-stealth] ([xshift=0.7em]subquo.east) -- ([xshift=-0.8em]learner.west);

\draw[-stealth] ([xshift=1em,yshift=1.15em]learner.north) to[bend left=27] ([xshift=-1em,yshift=1.3em]certifier.north);
\draw[-stealth] ([xshift=-1em,yshift=-1.2em]certifier.south) to[bend left=27] ([xshift=1em,yshift=-1.2em]learner.south);

\node at (11,10.8) {\scriptsize $V_{\theta}$};
\node at (11,8.45) {\scriptsize $(x_{cex},t_{cex})$};
\node at (11.1,7.8) {\footnotesize CEGIS};

\draw[-stealth] (1.8,9.625) -- (2.5,9.625);

\draw[-stealth] (13.5,9.625) -- (14.5,9.625);
\node[anchor=west] at (14.5,9.625) {\large $V_{\theta}, \epsilon$};

\end{tikzpicture}
}
\caption{Flowchart of the training and verification Procedure}
\label{fig:cegis}
\end{figure*}

\section{Case Studies}\label{sec6:case_study}
Given the page limitations, we prioritize detailing the key steps of the training and verification loop to ensure clarity, and thus omit a full comparison of different reach/avoid examples. The provided toolbox\footnote{\url{https://github.com/nikovert/CARe}} includes additional examples, such as an evader-pursuer scenario, forward and backward reachability problems, and implementations of both the reachable tube and the reachable set.

Let us consider the simple double integrator example, a canonical second-order control system, \(\dot{x}_1 = x_2, \dot{x}_2 = u\), where $x_1$ represents position, $x_2$ represents velocity, and $u$ is the control input. We consider the problem of learning the forward reachable set, with the initial set defined by
\begin{equation}
    g(x) = x_1^2 + x_2^2 - R^2,
\end{equation}
with \(R^2 = 0.5\).

\subsection{Training of \(V_{\theta}(t,x)\)}
The training procedure, outlined in Algorithm 1 and illustrated in Fig. \ref{fig:cegis} follows a curriculum learning strategy with three phases:
\begin{itemize}
    \item \textbf{Pretraining phase:} Focuses on boundary conditions at the final time \(t = T\).
    \item \textbf{Curriculum phase:} Gradually extends the time horizon from \([T, T]\) to \([t_{0}, T]\).
    \item \textbf{Finetuning phase:} Further refines the model to minimize loss before certification.
\end{itemize}
Each training epoch involves generating a random batch of samples and computing the loss \(\mathcal{L}(t_i, x_i, \theta)\) for each sample \(i \in \{1, \ldots, N\}\). The formal error bounds on the reachable set can be conservative as the results in Theorem \ref{theorem_viscosity_sol_hjiPde} assume the worst-case realisation of \(\epsilon_2\) along any trajectory. Consequently, even when the mean training loss is low, the existence of regions of the state space with high network loss can lead to conservative over- or under-approximations of the reachable set. To mitigate this, we not only minimize the mean error over the batch but also the maximum error, leading to the training objective of minimizing:
\begin{align*}
    \frac{1}{N} \sum_i \mathcal{L}(t_i, x_i, \theta) + \lambda_{\max} \max_{i} \mathcal{L}(t_i, x_i, \theta).
\end{align*}
Here, \(\lambda_{\max}\) is set to \(10\%\) during the curriculum phase and \(30\%\) during the finetuning phase. The finetuning phase uses a patience counter, terminating only when the loss consistently stays below a threshold \(\lambda_{\epsilon} \epsilon_2\), with \(\lambda_{\epsilon} = 95\%\). This modified approach to the curriculum training scheme presented in \cite{bansal2021deepreach} reduces the total training time while ensuring sufficient confidence in the model prior to progressing to the certification phase.

\begin{algorithm} \label{alg:1}
\caption{Training of $V_\theta(t,x)$}
\begin{algorithmic}[1]
\State \textbf{Input:} Thresholds $\epsilon_1$ and $\epsilon_2$
\State \textbf{Initialize:} Model $V_\theta$, $t_{\mathrm{current}} \gets T$
\While{Training}
    \State Generate samples $x_i$ and $t_i \in [t_{\mathrm{current}}, T]$
    \State Compute model output $\hat{y} = V_\theta(t_i, x_i)$
    \State Compute loss $\mathcal{L}_{total} = \mathcal{L}_{mean} + \lambda_{max} \cdot \mathcal{L}_{max}$
    \State Backpropagate $\mathcal{L}_{total}$ and update parameters $\theta$ 
    \If{$\mathcal{P}_{\mathrm{pre-training}}$ \textbf{and} $\text{max}(\mathcal{L}_{1}) < \epsilon_1$}
        \State Progress to $\mathcal{P}_{\mathrm{curriculum}}$
    \ElsIf{$\mathcal{P}_{\mathrm{curriculum}}$ \textbf{and} $\text{max}(\mathcal{L}_{2}) < \epsilon_2$}
        \State Expand time horizon: $t_{\mathrm{current}} \gets t_{\mathrm{current}} -\Delta_t$
        \If{$t_{\mathrm{current}} = 0$}
            \State Progress to $\mathcal{P}_{\mathrm{finetune}}$
            \State $p \gets 0$
        \EndIf
    \ElsIf{$\mathcal{P}_{\mathrm{finetune}}$ \textbf{and} $\text{max}(\mathcal{L}_{2}) < \lambda_{\epsilon}\epsilon_2$}
            \State $p \gets p + 1$
            \If{$p > 1000$}
                \State \textbf{break} \Comment{Stop training}
            \EndIf
    \EndIf
\EndWhile
\State \textbf{Output:} Trained model $V_\theta$
\end{algorithmic}
\end{algorithm}
To improve performance and enable smaller networks to approximate the value function effectively, we introduce a polynomial layer as the first layer of the network. This allows subsequent layers to operate not only on the original inputs \((t,x)\) but also on polynomial transformations, such as \((t^2, x^2)\).

\begin{figure}[h]
    \centering
    \includegraphics[width=\columnwidth]{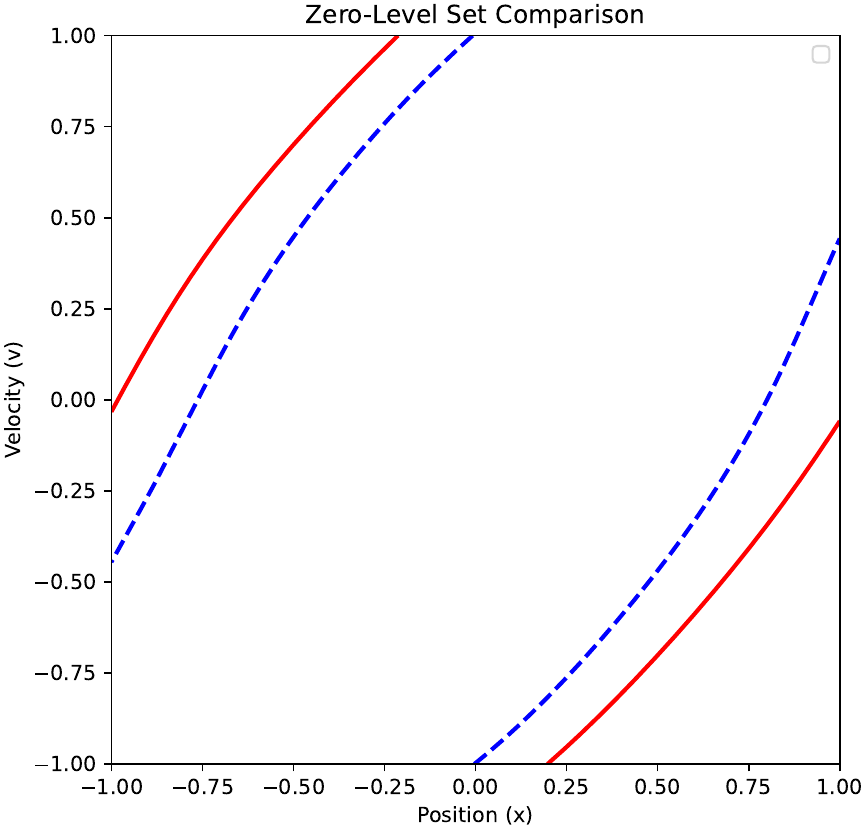}
    \caption{The over-approximated reachable set is shown in red, while the learned reachable set is depicted in blue. The initial set at time \(t_0 = 0\) is a circle around the centre of radius \(0.25\), the time horizon is \(T=1\).}
    \label{fig:double_int_trained}
\end{figure}

\subsection{Formal Verification}
To formally verify that the trained network represents an $\epsilon$-accurate value function, we generate a symbolic representation \(V_{\theta}(t, x)\) to instantiate the quantifier-free SMT query described in Section~\ref{sec4:smt}. To obtain the corresponding derivative terms \( D_{t} V_{\theta}(t, x) \) and \( \mathcal{H}(t, x, D_{x}V_{\theta}) \) we apply symbolic differentiation. To further parallelize the SMT call, we decompose the absolute value expressions of Equations \eqref{formula_1}--\eqref{formula_2} into separate calls, e.g., we seek $x$ s.t.
\begin{align*}
    & x \in \mathcal{X} \land V_{\theta}(T, x) - g(x) > \epsilon_1, \\
    \lor \; & x \in \mathcal{X} \land V_{\theta}(T, x) - g(x) < -\epsilon_1.
\end{align*}
As SMT with nonlinear real arithmetic is undecidable in general, no exact decision procedure can exist for arbitrary HJ-PDEs~\cite{DBLP:series/txtcs/KroeningS08}. However, \texttt{dreal} provides a $\delta$-complete decision procedure that allows for user-specified $\delta$ deviations in the satisfying assignments~\cite{DBLP:conf/cade/GaoAC12,gao2013dreal}. Consequently, while satisfying assignments may be spurious, unsatisfiability carries over to the exact problem. Hence, if the SMT query is unsatisfiable, it formally establishes a global upper bound on the loss function across the entire domain of interest.


\begin{table}[h]
    \centering
    \small
    \caption{Training and verification results of the CEGIS loop}
    \setlength{\tabcolsep}{4pt} 
    \renewcommand{\arraystretch}{1.1} 
    \begin{tabular}{|c|c|c|c|c|}
        \hline
        Iter. & $\epsilon$ & Training (s) & Verif. (s) & Result \\ 
        \hline
        1 & 0.30 & 836.92 & 1367& Certified \\ 
        2 & 0.27 & - & 22 & Counterexample found \\ 
        3 & 0.27 & 181.27 & 2735 & Certified \\ 
        4 & 0.243 & - & 5537 & Certified \\ 
        5 & 0.2187 & - & 237& Counterexample found \\
        \hline
    \end{tabular}
    \label{tab:results}
\end{table}
Checking SMT queries with nonlinear real arithmetic can be computationally expensive, especially as unsatisfiability constitutes a full proof for validity. In practice, if a counterexample exists, the SMT query typically terminates within a reasonable timeframe. 

For illustration, we consider training and verifying the double integrator system using a single hidden layer with 16 neurons, employing sine activation functions. After training with \(\epsilon \coloneqq (\epsilon_1 + \epsilon_2 (T-t_0))\), and  \(T=1, t_0=0, \epsilon_2 = 0.95\epsilon = 0.285, \epsilon_1 = 0.05\epsilon = 0.015\) we progress to certification, which takes 1367 seconds resulting in no counterexample being found. Subsequently, we attempt to reduce \( \epsilon \) by \(10\%\). Within 22 seconds, one of the parallelized SMT queries finds a counterexample. Thus, at this stage, our trained neural network satisfies \( \mathcal{L}(t,x) < 0.3 \) but not \( \mathcal{L}(t,x) < 0.27 \). Following a Counterexample-Guided Inductive Synthesis (CEGIS) approach, we leverage the identified counterexample to refine our model (Fig. \ref{fig:cegis}). Specifically, we re-enter the finetuning phase of training with the reduced \( \epsilon \) this time ensuring that \(10\%\) of our samples are drawn near the counterexample. This refinement improves the model such that a subsequent verification attempt fails to find a new counterexample, thereby establishing a tighter bound on the loss. We continue the CEGIS loop for a total of 5 iterations with the results shown in Table \ref{tab:results}. The fine-tuning phase of the final iteration is unable to obtain an empirical loss below \(0.2187\), thus, the best certifiable result with \(\epsilon = 0.243\) is used to produce the final over-approximated reachable set, shown in Fig. \ref{fig:double_int_trained}.

Our approach is applicable to arbitrary continuous-time reachability problems. However, the use of \texttt{dreal} practically limits the network size, as increasing the depth and number of neurons adds complexity to the formulas \eqref{formula_1}-\eqref{formula_2}. If an empirical result is satisfactory, Theorem \ref{theorem_viscosity_sol_hjiPde} can provide an over- or under-approximation of the reachable set, even when only an empirical bound on the loss is available.

\section{Conclusions}\label{sec7:conclusion}
The primary goal of this paper is to establish a relation between the training loss and the induced level of approximation error on the final BRS/BRT. To this end, we introduce a modified HJ-PDE allowing us to relate a bound on the loss function used during training, to the true solutions of the reachability problem. To obtain the bound on the loss function, we could use an empirical data-driven approach to estimating the bound on the training loss, though this would fail to constitute a formal certificate. Thus we introduce an approach to verification using a Satisfiability Modulo Theory solver, enabling us to provide a formal certificate for approximate reachability.
\iflongversion
    \appendix
For the proofs in the remaining sections, we ease notation by omitting the subscript of \(\epsilon\), i.e. \(\epsilon_2 = \epsilon\).  we introduce the payoff function:
\begin{align}\label{eq_main_payoff_function}
    P(u,d) &= P_{x,t} (u(\cdot), \beta(\cdot)) \nonumber \\
    &= \int_{t}^{T} \epsilon \, ds + g(\phi(T,t,x,u(\cdot),\beta(\cdot))),
\end{align}
where \( g:\mathbb{R}^m \to \mathbb{R} \) satisfies the following conditions:
\begin{align}
    |g(x)| &\leq C_{2}, \label{eq_main_final_cost_1} \\
    |g(x)-g(\hat{x})| &\leq C_{2} |x - \hat{x}|, \label{eq_main_final_cost_2}
\end{align}
for some constant \( C_{2} \) and for all \( t_{0} \leq t \leq T \), \( x, \hat{x} \in \mathbb{R}^{m} \), \( u \in \mathcal{U} \), and \( d \in \mathcal{D} \). In the differential game setup considered in this paper, the control \( u \) is chosen to maximize \( P(u,d) \), while the disturbance \( d \) is chosen to minimize \( P(u,d) \). The term \( \epsilon \) represents a positive perturbation \( (\epsilon \in \mathbb{R}^{+}) \). Since the proofs for \( \underline{V} \) follow directly from those of \( \overline{V} \), we focus only on the case of \( \overline{V} \) and ease notation by using \( V \) to denote \( \overline{V} \). We ease notation be introducting \(\mathbf{x}_{u,\beta}(\tau) \coloneqq \phi(\tau,t,x,u(\cdot),\beta(\cdot))\).

\subsection{Bellman Optimality}
\begin{thm}\label{theorem_belman_optimality}
    The value function defined in Equation~\eqref{eq_main_value_func} satisfies the Bellman optimality principle:
    \begin{align*}
        \overline{V}(t,x) &= \inf_{\beta(\cdot) \in \Delta_{[t,t+\sigma]}} \sup_{u(\cdot) \in \mathcal{M}_{[t,t+\sigma]}} 
        \min \biggl\{ \\
        & \overline{V}(t+\sigma, \mathbf{x}_{u,\beta}(t+\sigma)) + \int_{t}^{t+\sigma} \epsilon \, ds, \\
        & \inf_{\tau \in [t, t+\sigma]} \Bigl[\int_{t}^{\tau} \epsilon \, ds + g(\mathbf{x}_{u,\beta}(\tau))\Bigr] \biggr\}.
    \end{align*}
\end{thm}

\begin{proof}
We define 
\begin{align*}
    W(t,x) &= \inf_{\beta(\cdot) \in \Delta_{[t,t+\sigma]}} \sup_{u(\cdot) \in \mathcal{M}_{[t,t+\sigma]}} 
    \min \Bigg\{ \\
    & \quad V(t+\sigma, \mathbf{x}_{u,\beta}(t+\sigma)) + \int_{t}^{t+\sigma} \epsilon \, ds, \\
    & \quad \inf_{\tau \in [t, t+\sigma]} \left[ \int_{t}^{\tau} \epsilon \, ds + g(\mathbf{x}_{u,\beta}(\tau)) \right] \Bigg\}.
\end{align*}

\noindent \textbf{Part I: Lower bound of \( W(t,x) \):} \\
Fix \( \gamma > 0 \), and choose \( \beta_1 \in \Delta_{[t,t+\sigma]} \) such that:
\begin{align*}
    W(t,x) \geq & \sup_{u_1(\cdot) \in \mathcal{M}_{[t,t+\sigma]}} \min \Bigg\{ \\
    & V(t+\sigma,\mathbf{x}_{u_1,\beta_1}(t+\sigma)) + \int_{t}^{t+\sigma} \epsilon \, ds, \\
    & \hspace{-1cm} \inf_{\tau \in [t, t+\sigma]} \left[ \int_{t}^{\tau} \epsilon \, ds + g(\mathbf{x}_{u_1,\beta_1}(\tau)) \right] \Bigg\} - \gamma.
\end{align*}
Thus, for all \( u_1(\cdot) \in \mathcal{M}_{[t,t+\sigma]} \), there exists a \( \beta_1 \in \Delta_{[t,t+\sigma]} \) such that:
\begin{align} \label{eq_theorem_1_1}
    W(t,x) &\geq \min \Bigg\{ 
    V(t+\sigma, \mathbf{x}_{u_1,\beta_1}(t+\sigma)) +  \int_{t}^{t+\sigma} \epsilon \, ds, \nonumber \\
    & \inf_{\tau \in [t, t+\sigma]} \left[ \int_{t}^{\tau} \epsilon \, ds + g(\mathbf{x}_{u_1,\beta_1}(\tau)) \right] \Bigg\} - \gamma.
\end{align}

Now, consider the term \( V(t+\sigma,x) \):
\begin{align*}
    V(t+\sigma,x) &= \inf_{\beta(\cdot) \in \Delta_{[t+\sigma, T]}} \sup_{u(\cdot) \in \mathcal{M}_{[t+\sigma, T]}} \inf_{\tau \in [t+\sigma,T]} \Bigg[ \int_{t+\sigma}^{\tau} \epsilon \, ds \\
    & \quad + g(\phi(\tau, t+\sigma, x, u(\cdot), \beta(\cdot))) \Bigg].
\end{align*}
Similarly, we can find a \( \beta_2(\cdot) \in \Delta_{[t+\sigma, T]} \) and \( u_2(\cdot) \in \mathcal{M}_{[t+\sigma, T]} \) such that:
\begin{align*}
    V(t+\sigma, x) \geq &\inf_{\tau \in [t+\sigma,T]} \Bigg[ \int_{t+\sigma}^{\tau} \epsilon \, ds \\
    &+ g(\phi(\tau, t+\sigma,x,u_2(\cdot),\beta_2(\cdot))) \Bigg] - \gamma.
\end{align*}

Now, define the combined control policy \( u(\cdot) \in \mathcal{M}_{[t,T]} \) and disturbance strategy \( \beta(\cdot) \in \Delta_{[t,T]} \) as:
\begin{align*}
    u(\tau) &= \begin{cases}
        u_1(\tau) & \text{if } t \leq \tau < t + \sigma, \\
        u_2(\tau) & \text{if } t + \sigma \leq \tau \leq T.
    \end{cases} \\
    \beta[u](\tau) &= \begin{cases}
        \beta_1[u_1](\tau) & \text{if } t \leq \tau < t + \sigma, \\
        \beta_2[u_2](\tau) & \text{if } t + \sigma \leq \tau \leq T.
    \end{cases}
\end{align*}

Subsequently, we obtain:
\begin{align*}
    W(t,x) &\geq  \min \Big\{ \int_{t}^{t+\sigma} \epsilon \, ds + \\
    \inf_{\tau \in [t+\sigma,T]} &\left[ \int_{t+\sigma}^{\tau} \epsilon \, ds + g(\phi(\tau, t+\sigma, x, u_2(\cdot), \beta_2(\cdot))) \right], \\
    \inf_{\tau \in [t, t+\sigma]} &\left[ \int_{t}^{\tau} \epsilon \, ds + g(\phi(\tau, t, x, u_1(\cdot), \beta_1(\cdot))) \right] \Big\} - 2\gamma, \\
    \geq & \min \Big\{ \\
    \inf_{\tau \in [t+\sigma,T]} &\left[ \int_{t+\sigma}^{\tau} \epsilon \, ds + g(\phi(\tau, t+\sigma, x, u_2(\cdot), \beta_2(\cdot))) \right] \\
    \inf_{\tau \in [t, t+\sigma]} &\left[ \int_{t}^{\tau} \epsilon \, ds + g(\phi(\tau, t, x, u_1(\cdot), \beta_1(\cdot))) \right] \Big\} - 2\gamma, \\
    = & \inf_{\tau \in [t,T]} \left[ \int_{t}^{\tau} \epsilon \, ds + g(\phi(\tau, t, x, u(\cdot), \beta(\cdot))) \right] - 2\gamma.
\end{align*}

This holds true for all \( u(\cdot) \in \mathcal{M}_{[t,T]} \), and hence we conclude:
\begin{align*}
    W(t,x) \geq \sup_{u(\cdot) \in \mathcal{M}_{[t,T]}} \inf_{\tau \in [t,T]} \left[ \int_{t}^{\tau} \epsilon \, ds + g(\mathbf{x}_{u,\beta}(\tau)) \right] - 2\gamma.
\end{align*}
Finally, this gives:
\begin{align*}
    W(t,x) + 2\gamma \geq V(t,x).
\end{align*}

\noindent \textbf{Part II: Upper Bound \( W(t,x) \)}  \\
For all \( \beta_1(\cdot) \in \Delta_{[t,t+\sigma]} \), we have:
\begin{align*}
    W(t,x) \leq &\sup_{u(\cdot) \in \mathcal{M}_{[t,t+\sigma]}} \min\Bigl\{ \\ & V(t+\sigma, \mathbf{x}_{u,\beta_1}(t+\sigma)) + \int_{t}^{t+\sigma} \epsilon \, ds, \\
    &\inf_{\tau \in [t,t+\sigma]} \left[ \int_{t}^{\tau} \epsilon \, ds + g(\mathbf{x}_{u,\beta_1}(\tau)) \right] \Bigr\}.
\end{align*}
Then for a fixed \( \gamma > 0 \), there exists a \( u_1 \in \mathcal{M}_{[t,t+\sigma]} \) such that:
\begin{align*}
    W(t,x) &\leq \min \Bigl\{ V(t+\sigma, \mathbf{x}_{u_1,\beta_1}(t+\sigma)) + \int_{t}^{t+\sigma} \epsilon \, ds, \\
    &\quad \inf_{\tau \in [t,t+\sigma]} \left[ \int_{t}^{\tau} \epsilon \, ds + g(\mathbf{x}_{u_1,\beta_1}(\tau)) \right] \Bigr\} + \gamma.
\end{align*}
Consider the term
\begin{align*}
    V(t+\sigma, \mathbf{x}_{u,\beta}(t+\sigma)) &= \inf_{\beta \in \Delta_{[t+\sigma,T]}} \sup_{u(\cdot) \in \mathcal{M}_{[t+\sigma,T]}} \\
    \inf_{\tau \in [t+\sigma,T]} \Bigl[ \int_{t+\sigma}^{\tau} \epsilon \, ds &+ g(\phi(\tau, t+\sigma, x, u(\cdot), \beta(\cdot))) \Bigr].
\end{align*}
Then for all \( \beta_2 \in \Delta_{[t+\sigma,T]} \), we have:
\begin{align*}
    V(t+\sigma, x) &\leq \sup_{u(\cdot) \in \mathcal{M}_{[t+\sigma,T]}} \\ \inf_{\tau \in [t+\sigma,T]} &\Bigl[ \int_{t+\sigma}^{\tau} \epsilon \, ds + g(\phi(\tau, t+\sigma, x, u(\cdot), \beta_2(\cdot))) \Bigr].
\end{align*}
For a fixed \( \gamma > 0 \), there exists a \( u_2 \in \mathcal{M}_{[t+\sigma,T]} \) such that:
\begin{align*}
    V(t+\sigma, x) &\leq \inf_{\tau \in [t+\sigma,T]} \Bigl[ \int_{t+\sigma}^{\tau} \epsilon \, ds \\ 
    & + g(\phi(\tau, t+\sigma, x, u_2(\cdot), \beta_2(\cdot))) \Bigr] + \gamma.
\end{align*}
We define the combined control policy \( u(\cdot) \in \mathcal{M}_{[t,T]} \) and disturbance strategy \( \beta(\cdot) \in \Delta_{[t,T]} \) as before, such that:
\begin{align} \label{eq_theorem_1_2}
    W(t,x) &\leq \min \Bigl\{ \nonumber \\ 
    & \inf_{\tau \in [t+\sigma,T]} \Bigl[ \int_{t+\sigma}^{\tau} \epsilon \, ds + g(\phi(\tau, t+\sigma, x, u_2(\cdot), \beta_2(\cdot))) \Bigr], \nonumber \\
    &\quad \inf_{\tau \in [t,t+\sigma]} \Bigl[ \int_{t}^{\tau} \epsilon \, ds + g(\mathbf{x}_{u_1,\beta_1}(\tau)) \Bigr] \Bigr\} + 2\gamma. \nonumber \\
    &\leq \min \Bigl\{ \nonumber \\
    & \inf_{\tau \in [t+\sigma,T]} \Bigl[ \int_{t}^{\tau} \epsilon \, ds + g(\phi(\tau,t+\sigma,x,u_2(\cdot),\beta_2(\cdot))) \Bigr], \nonumber \\
    &\quad \inf_{\tau \in [t,t+\sigma]} \Bigl[ \int_{t}^{\tau} \epsilon \, ds + g(\mathbf{x}_{u,\beta_1}(\tau) \Bigr] \Bigr\} + 2\gamma. \nonumber \\
    &\leq \inf_{\tau \in [t,T]} \Bigl[ \int_{t}^{\tau} \epsilon \, ds + g(\mathbf{x}_{u,\beta}(\tau)) \Bigr] + 2\gamma.
\end{align}
It follows from Equation \eqref{eq_main_value_func} that:
\begin{align*}
    V(t,x) &= \inf_{\beta \in \Delta_{[t,T]}} \sup_{u \in \mathcal{M}_{[t,T]}} \inf_{\tau \in [t,T]} \Bigl[ \int_{t}^{\tau} \epsilon \, ds + g(\mathbf{x}_{u,\beta}(\tau)) \Bigr].
\end{align*}
Thus, there exists \( \beta \in \Delta_{[t,T]} \) such that for a given \( \gamma > 0 \):
\begin{align*}
    V(t,x) &\geq \sup_{u \in \mathcal{M}_{[t,T]}} \inf_{\tau \in [t,T]} \Bigl[ \int_{t}^{\tau} \epsilon \, ds + g(\mathbf{x}_{u,\beta}(\tau)) \Bigr] - \gamma.
\end{align*}
Thus for all \( u \in \mathcal{M}_{[t,T]} \), we get:
\begin{align}
    V(t,x) + \gamma &\geq \inf_{\tau \in [t,T]} \Bigl[ \int_{t}^{\tau} \epsilon \, ds + g(\mathbf{x}_{u,\beta}(\tau)) \Bigr].
\end{align}
Combining this with Equation \eqref{eq_theorem_1_2}, we get:
\begin{align}
    W(t,x) &\leq V(t,x) + 3\gamma.
\end{align}

Finally, since \( \gamma \) is arbitrary, we can let \( \gamma \to 0 \) and obtain:
\begin{align*}
    V(t,x) &= \inf_{\beta(\cdot) \in \Delta_{[t,t+\sigma]}} \sup_{u(\cdot) \in \mathcal{M}_{[t,t+\sigma]}} \min \Bigl\{ \\
    & \quad  V(t+\sigma, \mathbf{x}_{u,\beta}(t+\sigma)) + \int_{t}^{t+\sigma} \epsilon \, ds, \\
    &\quad \inf_{\tau \in [t,t+\sigma]} \Bigl[ \int_{t}^{\tau} \epsilon \, ds + g(\mathbf{x}_{u,\beta}(\tau)) \Bigr] \Bigr\}.
\end{align*}
\end{proof}

The proof of Theorem \ref{theorem_viscosity_sol_hjiPde} relies on several further properties of the value function, which will be provided in the form of Lemmas \ref{lemma_value_func_monotonic}--\ref{lemma_hamiltonian_1}. As in the previous section, we will use \( V \) for ease of notation. We begin with the monotonicity of the value function.
\subsection{Monotonicity of the value function}
\begin{lemma}\label{lemma_value_func_monotonic}
    For all $(t,x) \in [0,T]\times\R^{m}$
    \begin{align*}
        V(t,x) \leq V(t+\sigma,x) + \int_{t}^{t+\sigma}\epsilon ds \hfill
    \end{align*}
    and $V(T,x) = g(x)$
\end{lemma}
\begin{proof}
    \( V(T,x) = g(x) \) is trivial and can be directly observed from the definition of the value function, Equation \eqref{eq_main_value_func}. 
    Consider
    \begin{align*}
        V(t,x) =  \inf_{\beta \in \Delta_{[t,T]}} \sup_{u \in \mathcal{M}_{[t,T]}} \inf_{\tau \in [t,T]}
        \Bigl[\int_{t}^{\tau}\epsilon \, ds + g(\mathbf{x}_{u,\beta}(\tau))\Bigr]
    \end{align*}
    Let us assume for the sake of contradiction that
    \begin{align*}
        V(t,x) > V(t+\sigma,x) + \int_{t}^{t+\sigma}\epsilon \, ds,
    \end{align*}
    which can be equivalently stated as
    \begin{align*}
        &\inf_{\beta_{1} \in \Delta_{[t,T]}} \sup_{u_{1} \in \mathcal{M}_{[t,T]}} \inf_{\tau \in [t,T]} \Bigl[\int_{t}^{\tau}\epsilon \, ds + g(\mathbf{x}_{u_1,\beta_1}(\tau))\Bigr] \\ 
        &> \inf_{\beta_{2} \in \Delta_{[t+\sigma,T]}} \sup_{u_{2} \in \mathcal{M}_{[t+\sigma,T]}} \inf_{\tau \in [t+\sigma,T]} \Bigl[\int_{t+\sigma}^{\tau}\epsilon \, ds \\
        & \qquad + g(\phi(\tau,t+\sigma,x,u_2(\cdot),\beta_2(\cdot)))\Bigr] + \int_{t}^{t+\sigma}\epsilon \, ds.
    \end{align*}
    Thus, for all \( \beta_{1} \in \Delta_{[t,T]} \), there exists a \( \beta_{2} \in \Delta_{[t+\sigma,T]} \) such that 
    \begin{align*}
        & \sup_{u_{1} \in \mathcal{M}_{[t,T]}} \inf_{\tau \in [t,T]} \Bigl[\int_{t}^{\tau}\epsilon \, ds + g(\mathbf{x}_{u_1,\beta_1}(\tau))\Bigr] \\ 
        &> \sup_{u_{2} \in \mathcal{M}_{[t+\sigma,T]}} \inf_{\tau \in [t+\sigma,T]} \Bigl[\int_{t+\sigma}^{\tau}\epsilon \, ds \\
        & \qquad + g(\phi(\tau,t+\sigma,x,u_2(\cdot),\beta_2(\cdot)))\Bigr] + \int_{t}^{t+\sigma}\epsilon \, ds.
    \end{align*}
    Furthermore, for all \( u_{2} \in \mathcal{M}_{[t+\sigma,T]} \), there exists \( u_{1} \in \mathcal{M}_{[t,T]} \) such that 
    \begin{align*}
        & \inf_{\tau \in [t,T]} \Bigl[\int_{t}^{\tau}\epsilon \, ds + g(\mathbf{x}_{u_1,\beta_1}(\tau))\Bigr] \\ 
        &> \inf_{\tau \in [t+\sigma,T]} \Bigl[\int_{t+\sigma}^{\tau}\epsilon \, ds + g(\phi(\tau,t+\sigma,x,u_2(\cdot),\beta_2(\cdot)))\Bigr]  \\
        & \qquad + \int_{t}^{t+\sigma}\epsilon \, ds.
    \end{align*}
    Now let us define:
    \begin{align*}
        u(\tau) \equiv \begin{cases}
        u_{1}(\tau) & \text{if } t \leq \tau < t + \sigma, \\
        u_{1}(\tau) & \text{if } t + \sigma \leq \tau \leq T.
        \end{cases}
    \end{align*}
    and
    \begin{align*}
        \beta[u](\tau) \equiv \begin{cases}
        \beta_{1}[u](\tau) & \text{if } t \leq \tau < t + \sigma, \\
        \beta_{2}[u](\tau) & \text{if } t + \sigma \leq \tau \leq T.
        \end{cases}
    \end{align*}
    Using the uniqueness of the solution of the ODE \eqref{eq_main_ODE}, it follows that
    \begin{align*}
        & \inf_{\tau \in [t,T]} \Bigl[\int_{t}^{\tau}\epsilon \, ds + g(\mathbf{x}_{u,\beta}(\tau))\Bigr] \\ 
        &> \inf_{\tau \in [t+\sigma,T]} \Bigl[\int_{t+\sigma}^{\tau}\epsilon \, ds + g(\phi(\tau,t+\sigma,x,u(\cdot),\beta(\cdot)))\Bigr]  \\
        &> \inf_{\tau \in [t+\sigma,T]} \Bigl[\int_{t}^{\tau}\epsilon \, ds + g(\phi(\tau,t+\sigma,x,u(\cdot),\beta(\cdot)))\Bigr],
    \end{align*}
    which is a contradiction, as the infimum over a larger set is always less than or equal to the infimum over a smaller set.
\end{proof}

As an immediate corollary of Lemma \ref{lemma_value_func_monotonic}, we have
\begin{corollary}\label{coro_value_function_decreasing}
        For all $(t,x) \in [0,T]\times\mathbb{R}^{m}$
    \begin{align*}
        V(\tau,x(\tau)) + \inf_{\beta \in \Delta_{[t,T]}}\hspace{5mm} \sup_{u \in \mathcal{M}_{[t,T]}}  \int_{t}^{\tau} \epsilon ds \leq \\  \inf_{\beta \in \Delta_{[t,T]}}\hspace{5mm} \sup_{u \in \mathcal{M}_{[t,T]}}  (\int_{t}^{\tau} \epsilon ds + g(\mathbf{x}_{u,\beta}(\tau))
    \end{align*}
\end{corollary}

\subsection{Existence and Uniqueness}
We provide 2 lemmas that show that the value function is bounded and Lipschitz continuous. This is required for the existence and uniqueness guarantees.
\begin{lemma}\label{lemma_bounded_value_fnc}
    The value function, defined in the Equation \eqref{eq_main_value_func}, is bounded i.e. $|V(t,x)| \leq C_{4}$ Where $C_{4}$ is a constant 
\end{lemma}
\begin{proof}
    \begin{align*}
        P(u(\cdot),\beta(\cdot)) = \int_{t}^{T} \epsilon ds + g(\phi(T,t,x,u(\cdot),\beta(\cdot))) \hfill
    \end{align*}
    Using Equation \eqref{eq_main_final_cost_1} and the fact that $\epsilon$ is a constant, it follows that $|P(u(\cdot),\beta(\cdot))| \leq (T-t) \epsilon + C_{2}$. This holds for all $u(\cdot) \in \mathcal{M}_{[t,T]}$ and $\beta(\cdot) \in \Delta_{[t,T]}$. Thus implying $|V(t,x)| \leq C_{4}$.
\end{proof}

\begin{lemma}\label{lemma_lipschitz_value_fnc}
    The value function, defined in Equation \eqref{eq_main_value_func}, is Lipschitz continuous i.e. $|V(t,x)-V(\hat{t},\hat{x})| \leq C_{4}(|t-\hat{t}|+|x-\hat{x}|)$ for all $0 \leq t \leq T $, $0 \leq \hat{t} \leq T $ and $x,\hat{x} \in \mathbb{R}^{m}$.
\end{lemma}
\begin{proof}
    We introduce the following notation
    \begin{align*}
        P_{t,x}(\tau,u,\beta[u]) \coloneqq \int_{t}^{\tau} \epsilon ds + g(\mathbf{x}_{u,\beta}(\tau)) \hfill
    \end{align*}

    \textbf{Part I}\\
    We will show that $V(t_{1}, x_{1}) - V(t_{2}, x_{2})  \leq \overline{C} (|t_{1}-t_{2}| + |x_{1}-x_{2}|) + 3 \gamma$, where $\overline{C}$ is some constant. Without loss of generality we can assume that $t_{1} \leq t_{2}$. Furthermore, for any $\beta_{1}$
    \begin{align*}
        V(t_{1},x_{1}) = &\inf_{\beta \in \Delta_{[t_{1},T]}} \sup_{u \in \mathcal{M}_{[t_{1},T]}} \inf_{\tau \in [t_{1},T]} P_{t_{1},x_{1}}(\tau,u,\beta[u]) \\
        \leq &\sup_{u \in \mathcal{M}_{[t_{1},T]}} \inf_{\tau \in [t_{1},T]}  P_{t_{1},x_{1}}(\tau,u,\beta_{1}[u]) \hfill
    \end{align*}
    Subsequently, for a fixed $\gamma > 0$, $\exists u_{1}\in \mathcal{M}_{[t_{1},T]}$ such that for all $\tau \in [t_{1},T]$ we have
    \begin{align}\label{eq_lemma_lip_1}
        V(t_{1},x_{1}) \leq P_{t_{1},x_{1}}(\tau,u_{1},\beta_{1}[u_{1}]) + \gamma \hfill
    \end{align}

    Similarly, we can show that for a fixed $\gamma > 0$, $\exists \beta_{2}\in \Delta_{[t_{2},T]}$ such that 
    \begin{align}
        V(t_{2},x_{2}) \geq \sup_{u \in \mathcal{M}_{[t_{2},T]}} \inf_{\tau \in [t_{2},T]}  P_{t_{2},x_{2}}(\tau,u,\beta_{2}[u]) - \gamma
    \end{align}
    For an arbitrary $u_{1} \in \mathcal{M}_{[t_{2},T]}$ we have that
    \begin{align*}
        V(t_{2}, x_{2}) \geq \inf_{\tau \in [t_{1},T]}  P_{t_{2},x_{2}}(\tau,u_{1},\tilde{\beta}[u_{1}]) - \gamma \hfill
    \end{align*}
    where $\tilde{\beta}[u_{1}]$ is defined as
    \begin{align*}
        \tilde{\beta}[u_{1}](\tau) \equiv \begin{cases}
        \beta_{1}[u_{1}](\tau) & \mathrm{ if } \quad t_{1} \leq \tau < t_{2} \\
        \beta_{2}[u_{1}](\tau)& \mathrm{ if } \quad t_{2} \leq \tau \leq T
    \end{cases} \hfill
    \end{align*}
    Using the definition of infimum $\exists \tau \in [t_{1},T]$ such that 
    \begin{align} \label{eq_lemma_lip_2}
        V(t_{2}, x_{2}) & \geq P_{t_{2},x_{2}}(\tau, u_{1},\tilde{\beta}[u_{1}]) - 2 \gamma \nonumber \\
        \Rightarrow - V(t_{2}, x_{2}) & \leq - P_{t_{2},x_{2}}(\tau, u_{1},\tilde{\beta}[u_{1}]) + 2 \gamma \hfill
    \end{align}
    
    Now let $x_{1}(\cdot)$ be the solution for the time horizon $(t_{1} \leq s \leq T)$ of the following ODE
    \begin{align*}
        \begin{cases}
            \frac{\mathrm{d}x_{1}}{\mathrm{d}s} =  f(s, x_{1}(s), u_{1}(s), \tilde{\beta}[u_{1}](s)) \\ 
            x_{1}(t_{1}) = x_{1}
        \end{cases}
    \end{align*}
    Let $x_{2}(\cdot)$  be the solution for the time horizon $(t_{2} \leq s \leq T)$ of the ODE
    \begin{align*}
        \begin{cases}
            \frac{\mathrm{d}x_{2}}{\mathrm{d}s} = f(s, x_{2}(s), u_{1}(s), \beta_{2}[u_{1}](s)) \\
            x_{2}(t_{2}) = x_{2}
        \end{cases}
    \end{align*}

    Then it follows from equation \eqref{eq_lemma_lip_1} and \eqref{eq_lemma_lip_2}, that
    \begin{align}\label{eq_lemma_lip_3}
        V(t_{1}, x_{1}) - V(t_{2}, x_{2}) \leq& \nonumber \\
        P_{t_{1},x_{1}}(\tau, u_{1}, \beta_{1}[u_{1}]) - &P_{t_{2},x_{2}}(\tau, u_{1},\tilde{\beta}[u_{1}]) + 3 \gamma \nonumber\\
        = \int_{t_{1}}^{\tau}\epsilon ds + g(x_{1}(\tau))
        &- \int_{t_{2}}^{\tau}\epsilon ds - g(x_{2}(\tau)) + 3 \gamma \nonumber \\
        = g(x_{1}(\tau)) - g(x_{2}(\tau)) & + 3 \gamma + (t_{2} - t_{1}) \epsilon
    \end{align}

    Using Equation \eqref{eq_main_final_cost_2}, it follows that 
    \begin{equation}
        \big| g(x_{1}(\tau)) - g(x_{2}(\tau)) \big| \leq C_2\big| x_{1}(\tau) - x_{2}(\tau) \big|
    \end{equation}
    
    Furthermore, since $\tilde{\beta}[u_{1}](s) = \beta_{2}[u_{1}](s)$ for all $s \in [t_2, T]$, we have that
    \begin{align*}
        |x_{1}(\tau) &- x_{2}(\tau)|
        = \Big| x_{1} + \int_{t_{1}}^{\tau} f(s, x_{1}(s), u_{1}(s), \tilde{\beta}[u_{1}](s)) ds \\
        & \quad - x_{2} - \int_{t_{2}}^{\tau} f(s, x_{2}(s), u_{1}(s), \beta_{2}[u_{1}](s)) ds \Big| \\
        & = \Big| x_{1} + \int_{t_{1}}^{t_{2}} f(s, x_{1}(s), u_{1}(s), \tilde{\beta}[u_{1}](s)) ds - x_{2}\\
        & \quad + \int_{t_{2}}^{\tau} f(s, x_{1}(s), u_{1}(s), \beta_{2}[u_{1}](s)) \\
        & \quad - f(s, x_{2}(s), u_{1}(s), \beta_{2}[u_{1}](s))ds \Big| \\
        & \leq \Big| x_{1}(t_{2}) - x_{2}(t_{2}) \Big| \\
        & \quad + \int_{t_{2}}^{\tau} \Big| f(s, x_{1}(s), u_{1}(s), \beta_{2}[u_{1}](s)) \\
        & \quad - f(s, x_{2}(s), u_{1}(s), \beta_{2}[u_{1}](s)) \Big| ds  \\
        & \leq \Big| x_{1}(t_{2}) - x_{2}(t_{2}) \Big| + C_1 \int_{t_{2}}^{\tau}  |x_1(s) - x_2(s)| ds \\
        & \leq \Big| x_{1}(t_{2}) - x_{2}(t_{2}) \Big| e^{(\tau-t_2)C_1}
    \end{align*}
    where the second inequality uses Equations \eqref{eq_main_system_assumtion_1} and \eqref{eq_main_system_assumtion_2} and the last inequality is due to the Bellman-Gronwall Lemma~\cite{Sastry_1999}.

    Using Equation \eqref{eq_main_system_assumtion_1}, the following holds for any $u \in \mathcal{U}, d \in \mathcal{D}$
    \begin{align}\label{eq_lemma_lip_4}
        |x_{1}(t_{2}) - x_{1}| \leq C_{3} |t_{1} - t_{2}|,
    \end{align}
    \begin{align*}
    |x_{1}(t_{2}) - x_{1}(t_{1})| 
    &= \Big| \int_{t}^{t_{2}} f(s, x_{2}(s), u(s), d(s)) \, ds \\
    & \quad - \int_{t}^{t_{1}} f(s, x(s), u(s), d(s)) \, ds \Big| \\
    &\leq \left| C_{1}(t_{2}-t) - C_{1}(t_{1}-t) \right| \\
    &\leq C_{1} |t_{1} - t_{2}|.
    \end{align*} 

    Thus we obtain
    \begin{align*}
        \big| g(x_{1}(\tau)) &- g(x_{2}(\tau)) \big| \leq C_2\big| x_{1}(\tau) - x_{2}(\tau) \big| \\
        & \leq C_2e^{(\tau-t_2)C_1} | x_{1}(t_{2}) - x_{2}(t_{2}) | \\
        & \leq C_2e^{(\tau-t_2)C_1} | x_{1}(t_{2}) - x_{1} |  \\
        & \quad + C_2e^{(\tau-t_2)C_1} |x_{1} - x_{2}|\\
        & \leq C_{3} e^{(\tau-t_2)C_1} |t_{1} - t_{2}|  +  e^{(\tau-t_2)C_1} |x_{1} - x_{2}| 
    \end{align*}
    Subsequently, there exists some constant $\overline{C}$, such that
    \begin{align}\label{eq_lemma_lip_6}
        V(t_{1}, x_{1}) - V(t_{2}, x_{2}) \leq \overline{C} (|t_{1}-t_{2}| + |x_{1}-x_{2}|) + 3 \gamma
    \end{align}
    \textbf{Part II}\\
    Now we will show that $V(t_{2}, x_{2}) - V(t_{1}, x_{1}) \leq \overline{C} (|t_{1}-t_{2}| + |x_{1}-x_{2}|) +3 \gamma$. We have that
    \begin{align*}
        V(t_{2}, x_{2}) =  \inf_{\beta \in \Delta_{[t_{2},T]}} \sup_{u \in \mathcal{M}_{[t_{2},T]}} \inf_{\tau \in [t_{2},T]} P_{t_{2},x_{2}}(\tau,u,\beta[u]) 
    \end{align*}
    Thus for all  $\beta \in \Delta_{[t_{2},T]}$.
    \begin{align*}
        V(t_{2}, x_{2}) \leq \sup_{u \in \mathcal{M}(t_{2})}  \inf_{\tau \in [t_{2},T]} P_{t_{2},x_{2}}(\tau,u,\beta[u])
    \end{align*}
    Subsequently, for a fixed $\gamma > 0$, $\exists u_{2}\in \mathcal{M}_{[t_{2},T]}$ such that for all $\tau \in [t_{2},T]$ we have
    \begin{align}\label{eq_lemma_lip_7}
        V(t_{2},x_{2}) \leq P_{t_{2},x_{2}}(\tau,u_{2},\beta[u_{1}]) + \gamma \hfill
    \end{align}

    Similarly, we can show that for a fixed $\gamma > 0$, $\exists \beta_{1}\in \Delta_{[t_{1},T]}$ such that
    \begin{align}\label{eq_lemma_lip_8}
        V(t_{1},x_{1}) \geq \sup_{u \in \mathcal{M}_{[t_{1},T]}} \inf_{\tau \in [t_{1},T]}  P_{t_{1},x_{1}}(\tau,u,\beta_{1}[u]) - \gamma
    \end{align}
    For an arbitrary $u_{1} \in \mathcal{M}_{[t_{2},T]}$ we have that
    \begin{align*}
        V(t_{1}, x_{1}) \geq \inf_{\tau \in [t_{1},T]}  P_{t_{1},x_{1}}(\tau,u_{1},\beta_{1}[u_{1}]) - \gamma \hfill
    \end{align*}
    
    Using the definition of infimum, $\exists \tau \in [t_{1},T]$ such that 
    \begin{align} \label{eq_lemma_lip_9}
        V(t_{1}, x_{1}) & \geq P_{t_{1},x_{1}}(\tau, u_{1},\beta_{1}[u_{1}]) - 2 \gamma \hfill \nonumber \\
        \Rightarrow - V(t_{1}, x_{1}) & \leq - P_{t_{1},x_{1}}(\tau, u_{1},\beta_{1}[u_{1}]) + 2 \gamma \hfill
    \end{align}
    Let us define the extended policy
    \begin{align*}
        \tilde{u}(\tau) \equiv \begin{cases}
        u_{1}(\tau) & \mathrm{ if } \quad t_{1} \leq \tau < t_{2} \\
        u_{2}(\tau)& \mathrm{ if } \quad t_{2} \leq \tau \leq T
    \end{cases} \hfill
    \end{align*}
    Now by using Equation \eqref{eq_lemma_lip_7} and \eqref{eq_lemma_lip_9}, it follows that
    \begin{align}\label{eq_lemma_lip_10}
        V(t_{2}, x_{2}) - V(t_{1}, x_{1}) \leq& \nonumber \\
        P_{t_{2},x_{2}}(\tau,\tilde{u},\beta_{1}[\tilde{u}]) - &P_{t_{1},x_{1}}(\tau, \tilde{u},\beta_{1}[\tilde{u}]) + 3 \gamma
    \end{align}
    Following similar arguments as in Part 1 and since we can choose $\gamma$ to be arbitrarily small, we obtain $|V(t,x)-V(\hat{t},\hat{x})| \leq C_{4}(|t-\hat{t}|+|x-\hat{x}|)$.
\end{proof}

Next, let us recall the definition of the modified Hamiltonian
\begin{equation}
    \mathcal{H}(t,x,p) = \max_{u \in \mathcal{U}} \min_{d \in \mathcal{D}} [f(t, x, u, d) \cdot p + \epsilon].
\end{equation}

Then the final Lemma required for Theorem \ref{theorem_viscosity_sol_hjiPde} is
\begin{lemma} \label{lemma_hamiltonian_1}
    Let $\varphi \in C^{1}([t_{0}, T] \times \R^m)$ and $\theta > 0$. Then if $\varphi$ satisfies
    \begin{equation}
        D_{t} \varphi (t_{0}, x_{0}) + \mathcal{H}(t_{0}, x_{0}, D_{x} \varphi (t_{0}, x_{0})) \leq -\theta \leq 0.
    \end{equation}
    then for a small enough $\sigma > 0$, there exists $ \beta \in \Delta_{[t_{0}, t_{0} + \sigma]}$ such that for all $u \in \mathcal{M}_{[t_{0}, t_{0} + \sigma]}$
    \begin{align*}
        \varphi ((t_{0} + \sigma), \phi(t_{0} + \sigma,t_{0},x_{0},u(\cdot),\beta(\cdot))) - \varphi (t_{0}, x_{0}) \\
        + \int_{t_{0}}^{t_{0}+\sigma} \epsilon ds  \leq -\frac{\theta \sigma}{2}
    \end{align*}
    Conversely, if $\varphi$ satisfies
    \begin{equation}
        D_{t} \varphi (t_{0}, x_{0}) + \mathcal{H}(t_{0}, x_{0}, D_{x} \varphi (t_{0}, x_{0})) \geq \theta \geq 0
    \end{equation}
    then for a small enough $\sigma > 0$, there exists $u \in \mathcal{M}_{[t_{0}, t_{0} + \sigma]}$ such that for all $ \beta \in \Delta_{[t_{0}, t_{0} + \sigma]}$
    \begin{align*}
        \varphi ((t_{0} + \sigma), \phi(t_{0} + \sigma,t_{0},x_{0},u(\cdot),\beta(\cdot))) - \varphi (t_{0}, x_{0}) \\
        + \int_{t_{0}}^{t_{0}+\sigma} \epsilon ds  \geq \frac{\theta \sigma}{2}.
    \end{align*}
\end{lemma}
\begin{proof}
    \textbf{Part I}  we define
    \begin{align*}
        \Lambda(t, x, u, d) = D_{t} \varphi(t, x) + f(t, x, u, d) \cdot D_{x} \varphi(t,x) + \epsilon.
    \end{align*}
    Then if $\max_{u \in \mathcal{U}}  \min_{d \in \mathcal{D}} \Lambda (t_{0}, x_{0}, u, d) \leq - \theta < 0$ then for each $u \in \mathcal{U}$ there exists $d = d(u) \in \mathcal{D}$ such that $\Lambda (t_{0}, x_{0}, u, d) \leq - \theta$. Since $\Lambda$ is uniformly continuous we have
    \begin{align*}
        \Lambda(t_{0}, x_{0},\tilde{u}, d) \leq  \frac{-3\theta}{4} \hfill
    \end{align*}
    for all $\tilde{u} \in \mathcal{B}(u, r) \cap \mathcal{U}$ and some $r = r(u) > 0$. Since $\mathcal{U}$ is compact there exist finitely many distinct points $u_{1}, u_{2}, ... u_{n} \in \mathcal{U}$, $d_{1}, d_{2}, ... d_{n} \in \mathcal{D}$ and $r_{1}, r_{2}, ... r_{n} > 0$ such that 
    \begin{align*}
        \mathcal{U} \subset \bigcup_{i=1}^{n} B (u_{i}, r_{i}) \hfill
    \end{align*}
    and 
    \begin{align*}
        \Lambda(t_{0}, x_{0},\tilde{u}, d_{i}) \leq  \frac{-3\theta}{4} \quad  \mathrm{ for } \quad \tilde{u} \in \mathcal{B}(u_{i}, r_{i}) \hfill
    \end{align*}
    We define \( \beta_{1} : \mathcal{U} \rightarrow \mathcal{D} \), setting
    \begin{align*}
        \beta_{1}(u) = d_{k} \: \mathrm{if} \: u\in \mathcal{B}(u_{k}, r_{k})\setminus \bigcup_{i=1}^{k-1} \mathcal{B}(u_{i}, r_{i}) \quad (k=1,\ldots, n).
    \end{align*} 
    Thus $\Lambda(t_{0}, x_{0}, u, \beta_{1}(u)) \leq  \frac{-3\theta}{4}$ for all $u \in \mathcal{U}$. Since $\Lambda$ is uniformly continuous we therefore have for each sufficiently small $\delta > 0$
    \begin{align*}
        \Lambda(s, x(s), u, \beta_{1}(u)) \leq  \frac{-\theta}{2} \hfil
    \end{align*}
    for all $u \in \mathcal{U}, t_{0} \leq s \leq t_{0} + \delta$ and any solution $x(\cdot)$ of Equation \eqref{eq_main_ODE} on $(t_{0}, t_{0} + \delta)$ for any $d(\cdot)$, $u(\cdot)$ with initial condition $x(t_{0}) = x_{0}$. 
    
    Finally we define $\beta \in \Delta_{[t_{0},T]}$ in the following way: 
    \begin{align*}
        \beta[u](s) = \beta_{1}(u(s)) \hfill
    \end{align*}
    for each $u \in \mathcal{M}_{[t_{0},T]}$. It then follows that 
    \begin{align*}
        \Lambda(s, x(s), u(s), \beta[u](s)) \leq  \frac{-\theta}{2} \quad  (t_{0} \leq s \leq t_{0} + \sigma),
    \end{align*}
    for each $u \in \mathcal{M}_{[t_{0},T]}$. Notice that
    \begin{align}
        \varphi(t_{0} + \sigma, x(t_{0} + \sigma)) = \varphi(t_{0}, x_{0}) \nonumber \\ + \int_{t_{0}}^{t_{0} + \sigma} f(s, x(s), u(s), \beta[u](s)) \cdot D_{x} \varphi(s,x(s)) \nonumber \\
        + D_{t} \varphi(s, x(s)) ds,
    \end{align}
    such that integrating over $\Lambda$ yields
    \begin{align*}
        \varphi(t_{0} + \sigma, x(t_{0} + \sigma)) - \varphi(t_{0}, x_{0}) 
        + \int_{t_{0}}^{t_{0} + \sigma} \epsilon ds \leq  - \frac{-\theta\sigma}{2}
    \end{align*}
    \textbf{Part II} Set
    \begin{align*}
        \Lambda(t, x, u, d) = D_{t} \varphi(t, x) + f(t, x, u, d) \cdot D_{x} \varphi(t,x) + \epsilon.
    \end{align*}
    Then if
    \begin{align*}
        \max_{u \in \mathcal{U}}  \min_{d \in \mathcal{D}} \Lambda (t_{0}, x_{0}, u, d) \geq \theta >0, \hfill
    \end{align*}
    there exists a $u^{*} \in \mathcal{U}$ such that 
    \begin{align*}
        \min_{d \in \mathcal{D}} \Lambda (t_{0}, x_{0}, u^{*}, d). \geq \theta \hfill
    \end{align*}
    Since $\Lambda$ is uniformly continuous, we have that
    \begin{align*}
        \Lambda(s, x(s), u^{*}, d) \geq \frac{\theta}{2} \hfill
    \end{align*}
    provided $ t_{0} \leq s \leq t_{0}+\delta$ (for any small $\delta > 0$) and $x(\cdot)$ solves ODE on $(t_{0}, t_{0} + \delta)$ for any $u(\cdot)$, $d(\cdot)$ with initial candidates $x(t_{0}) = x_{0}$. Hence for $u(\cdot) \equiv u^{*}$ and any and $\beta \in \Delta_{[t_{0},T]}$
    \begin{align*}
        D_{t} \varphi(s, x(s)) + f(s, x(s), u(s), \beta[u](s)) \cdot D_{x} \varphi(s,x(s)) \\
        + \epsilon \geq \frac{\theta}{2}
    \end{align*}

    Integrating this expression form $t_{0}$ to $t_{0} + \sigma$ and subtracting $\varphi(t_{0},x_{0})$ we obtain 
    \begin{align*}
        \varphi ((t_{0} + \sigma), x(t_{0} + \delta)) - \varphi (t_{0}, x_{0})
        + \int_{t_{0}}^{t_{0}+\sigma} \epsilon ds  \geq \frac{\theta \sigma}{2} \hfill
    \end{align*}
\end{proof}

We are now in a position to provide the proof of Theorem \ref{theorem_viscosity_sol_hjiPde}
\subsection{Proof of Theorem \ref{theorem_viscosity_sol_hjiPde}}
\begin{proof}
    To prove this it is sufficient to prove the following (\cite{choi2021robust}, \cite{mitchell2001validating}, \cite{tomlin2001safety})
    \begin{enumerate}
        \item For $\varphi(t,x) \in C^{1}([0, T] \times \R^m)$ such that $V-\varphi$ attains a local maximum at $(t_{0},x_{0}) \in [0, T]\times\R^{m}$ then
        \begin{align*}
            D_{t} \varphi (t_{0},x_{0}) +  \min \{\epsilon ,
        \mathcal{H}(t_{0}, x_{0}, D_{x} \varphi (t_{0}, x_{0}))\} \geq 0 \hfill
        \end{align*} \vspace{2mm}
        
        \item For $\varphi(t,x) \in C^{1}([0, T] \times \R^m)$ such that $V - \varphi$ has a local minimum at $(t_{0},x_{0}) \in [0, T]\times\R^{m}$ then 
        \begin{align*}
        D_{t} \varphi (t_{0},x_{0}) +  \min \{\epsilon ,\mathcal{H}(t_{0}, x_{0}, D_{x} \varphi (t_{0}, x_{0}))\} \leq 0 \hfill
        \end{align*}
    \end{enumerate}

    \textbf{Part 1}\\
    Let $\varphi(t,x) \in C^{1}([0, T] \times \R^m)$ and suppose that $V-\varphi$ attains a local maximum at $(t_{0},x_{0}) \in [0, T]\times\mathbb{R}^{m}$. Let us assume, for the sake of contradiction, that $\exists \theta > 0$ such that 
    \begin{align*}
        D_{t} \varphi (t_{0},x_{0}) +  \min \{\epsilon ,\mathcal{H}(t_{0}, x_{0}, D_{x} \varphi (t_{0}, x_{0}))\} \leq -\theta \leq 0\hfill
    \end{align*}
    \textbf{Case I:} $\mathcal{H}(t_{0}, x_{0}, D_{x} \varphi (t_{0}, x_{0})) < \epsilon$. \\
    \begin{align*}
        D_{t} \varphi (t_{0},x_{0}) + \mathcal{H}(t_{0}, x_{0}, D_{x} \varphi (t_{0}, x_{0}))  \leq -\theta \hfill
    \end{align*}
    According to the Lemma \ref{lemma_hamiltonian_1}, this implies
    \begin{align}
    \label{eqn:varphi_lemma4}
        \varphi(t_{0} + \sigma, x(t_{0} + \sigma)) - \varphi (t_{0}, x_{0}) 
        + \int_{t_{0}}^{t_{0}+\sigma} \epsilon ds \leq \frac{-\theta \sigma}{2}
    \end{align}
    Since we know that $V-\varphi$ has a maximum at $(t_{0}, x_{0})$, it follows that $V(t_{0}, x_{0}) - \varphi(t_{0}, x_{0}) \geq V(t_{0} + \sigma, x_{0}(t_{0}+ \sigma)) - \varphi(t_{0} + \sigma, x_{0}(t_{0}+ \sigma))$, such that, using the result of Equation \eqref{eqn:varphi_lemma4},
    \begin{align*}
        V(t_{0} + \sigma, x_{0}(t_{0}+ \sigma)) + \int_{t_{0}}^{t_{0}+\sigma} \epsilon ds + \frac{\theta \sigma}{2} \leq  V(t_{0}, x_{0})
    \end{align*}
    But this leads to a contradiction, since we know from Theorem \ref{theorem_belman_optimality} that 
    \begin{align*}
        V(t, x) \leq V(t + \sigma, x(t + \sigma)) + \int_{t}^{t+\sigma} \epsilon ds.
    \end{align*}

    \textbf{Case II:} $\epsilon \leq \mathcal{H}(t_{0}, x_{0}, D_{x} \varphi (t_{0}, x_{0}))$ \\
    \begin{align*}
        D_{t} \varphi (t_{0},x_{0}) +  \epsilon \leq -\theta
    \end{align*}
    Integrating the above from $t_{0}$ to $t_{0} + \sigma$, it follows that
    \begin{align}
        \varphi(t_{0} + \sigma, x_{0}) - \varphi (t_{0}, x_{0}) 
        + \int_{t_{0}}^{t_{0}+\sigma} \epsilon ds \leq -\theta \sigma
    \end{align}
    
    Since we know that $V-\varphi$ has a maximum, following similar argumentation as before, 
    \begin{align*}
        V(t_{0} + \sigma, x_{0}) + \int_{t_{0}}^{t_{0}+\sigma} \epsilon ds + \theta \sigma \leq  V(t_{0}, x_{0})
    \end{align*}
    But from Lemma \ref{lemma_value_func_monotonic} we know that 
    \begin{align*}
        V(t, x) \leq V(t + \sigma, x) + \int_{t_{0}}^{t_{0}+\sigma} \epsilon ds.
    \end{align*}
    This contradiction implies
    \begin{align*}
            D_{t} \varphi (t_{0},x_{0}) +  \min \{\epsilon ,
        \mathcal{H}(t_{0}, x_{0}, D_{x} \varphi (t_{0}, x_{0}))\} \geq 0 \hfill
    \end{align*}

    \textbf{Part 2}\\
    Let $\varphi(t,x) \in C^{1}([0, T] \times \R^m)$ and suppose that $V-\varphi$ attains a local minimum at $(t_{0},x_{0}) \in [0, T]\times\mathbb{R}^{m}$. We assume, for the sake of contradiction, that $\exists \theta > 0$ such that 
    \begin{align}
        \label{eqn:thm2_part2}
        D_{t} \varphi (t_{0},x_{0}) +  \min \{\epsilon, \mathcal{H}(t_{0}, x_{0}, D_{x}V)\} \geq \theta.
    \end{align}
    This implies 
    \begin{align*}
        D_{t} \varphi (t_{0},x_{0}) + \epsilon \geq \theta \hfill
    \end{align*}
    such that 
    \begin{align*}
        \varphi(t_{0} + \sigma, x(t_{0} + \sigma)) - \varphi (t_{0}, x_{0}) 
        + \int_{t_{0}}^{t_{0}+\sigma} \epsilon ds \geq \theta \sigma \geq \frac{\theta \sigma}{2}.
    \end{align*}
    Furthermore, Equation \eqref{eqn:thm2_part2} implies
    \begin{align*}
        D_{t} \varphi (t_{0},x_{0}) + \mathcal{H}(t_{0}, x_{0}, D_{x}V) \geq \theta.
    \end{align*}
     such that by Lemma \ref{lemma_hamiltonian_1}, we obtain
    \begin{align*}
        \varphi ((t_{0} + \sigma), x(t_{0} + \sigma)) - \varphi (t_{0}, x_{0}) + 
        \int_{t_{0}}^{t_{0}+\sigma} \epsilon ds  \geq \frac{\theta \sigma}{2}
    \end{align*}
    Since $V - \varphi$ obtains a local minimum at $(t_{0}, x_{0})$, i.e.
    \begin{align*}
        V(t_{0}, x_{0}) - \varphi(t_{0}, x_{0})
        \leq \\ V(t_{0} + \sigma, x(t_{0}+\sigma)) - \varphi(t_{0} + \sigma, x(t_{0}+\sigma)).
    \end{align*}
    it follows that
    \begin{align}\label{eq_part2_case_1_1.1}
        V(t_{0} + \sigma, x(t_{0}+\sigma)) - V(t_{0}, x_{0}) + \int_{t_{0}}^{t_{0}+\sigma} \epsilon ds \geq \frac{\theta \sigma}{2}.
    \end{align}
    Following the same logic and we also obtain
    \begin{align}\label{eq_part2_case_1_1.2}
        V(t_{0} + \sigma, x_{0}) - V(t_{0}, x_{0}) + \int_{t_{0}}^{t_{0}+\sigma} \epsilon ds \geq \frac{\theta \sigma}{2}.
    \end{align}
    
    Recall the definition of $V(t,x)$ obtained from Theorem \ref{theorem_belman_optimality}. We consider two cases \\
    \textbf{Case I}
    \begin{align*}
        V(t_{0} + \sigma, x(t_{0}+\sigma)) + \int_{t_{0}}^{t_{0} +\sigma} \epsilon ds \\
        \leq \inf_{\tau \in [t_{0}, t_{0} + \sigma]} 
        \Big[\int_{t_{0}}^{\tau} \epsilon ds
        + g(x(\tau))\Big]
    \end{align*}
    then
    \begin{align*}
        V(t_{0},x_{0}) = V(t_{0} + \sigma, x(t_{0}+\sigma)) + \int_{t_{0}}^{t_{0} +\sigma} \epsilon ds 
    \end{align*}
    substituting this in equation \eqref{eq_part2_case_1_1.1} leads to a contradiction.

    \textbf{Case II}
    \begin{align*}
        \inf_{\tau \in [t_{0}, t_{0} + \sigma]} 
        \Big[\int_{t_{0}}^{\tau} \epsilon ds
        + g(x(\tau))\Big] \leq \\ V(t_{0} + \sigma, x(t_{0}+\sigma)) + \int_{t_{0}}^{t_{0} +\sigma} \epsilon ds 
    \end{align*}
    Then
    \begin{align*}
        V(t_{0},x_{0}) = \inf_{\beta(\cdot) \in \Delta_{[t_{0},t_{0}+\sigma]}} \sup_{\nu(\cdot) \in \mathcal{M}_{[t_{0},t_{0}+\sigma]}} \\ \inf_{\tau \in [t_{0}, t_{0} + \sigma]} 
        \Big[\int_{t_{0}}^{\tau} \epsilon ds
        + g(x(\tau))\Big]
    \end{align*}
    The minima in $\tau$ must occur at $\tau = t_{0}$ and the minimizer is unique. If this was not true then there would exist some $\tau \in [t_{0}, t_{0}+\sigma]$ such that 

    \begin{align*}
        V(t_{0},x_{0}) = \inf_{\beta(\cdot) \in \Delta_{[t_{0},t_{0}+\sigma]}} \sup_{u(\cdot) \in \mathcal{M}_{[t_{0},t_{0}+\sigma]}} \\ 
        \Big[\int_{t_{0}}^{\tau} \epsilon ds
        + g(x(\tau))\Big]
    \end{align*}

    Thus using corollary \ref{coro_value_function_decreasing} we get that 
    \begin{align*}
        V(\tau,x(\tau))+ \int_{t_{0}}^{\tau} \epsilon ds  \leq V(t_{0},x_{0}), \quad \tau \in [t_{0}, t_{0}+\sigma]
    \end{align*}
    This will lead to contradiction with equation \eqref{eq_part2_case_1_1.1}. Furthermore, we know that for all $\tau \in [t_{0}, t_{0}+ \sigma]$ using Lemma \ref{lemma_value_func_monotonic}
    \begin{align*}
        V(t_{0},x_{0}) \leq  V(t_{0}+\sigma,x_{0})+ \\ \inf_{\beta(\cdot) \in \Delta_{[t_{0},t_{0}+\sigma]}}\quad  \sup_{u(\cdot) \in \mathcal{M}_{[t_{0},t_{0}+\sigma]}}\int_{t_{0}}^{t_{0}+\sigma}\epsilon ds  \leq g(x_{0})
    \end{align*}
    Therefore we have 
    \begin{align*}
        V(t_{0},x_{0}) =  V(t_{0}+\sigma,x_{0})+ \\ \inf_{\beta(\cdot) \in \Delta_{[t_{0},t_{0}+\sigma]}}\quad  \sup_{u(\cdot) \in \mathcal{M}_{[t_{0},t_{0}+\sigma]}}\int_{t_{0}}^{t_{0}+\sigma}\epsilon ds = g(x_{0})
    \end{align*}
    This along with equation \eqref{eq_part2_case_1_1.2} leads to contradiction.
\end{proof}
\fi

\bibliographystyle{abbrv}        
\bibliography{bibliography}           

\end{document}